\documentclass[submission,copyright,creativecommons]{eptcs}
\providecommand{\event}{AFL 2023} 

\usepackage{booktabs}
\usepackage{multirow}
\usepackage{amssymb, graphicx, amsmath, amsthm}

\newtheorem{theorem}{Theorem}
\newtheorem{example}[theorem]{Example} \newtheorem{proposition}[theorem]{Proposition} 
\newtheorem{corollary}[theorem]{Corollary} 

\newcommand{\eps}{\varepsilon}

\usepackage{tikz}
\usepackage{wrapfig}
\usetikzlibrary{arrows,automata,shapes.geometric}
\tikzset{every state/.style={minimum size=30}}
\usetikzlibrary{shapes.geometric}

\usepackage{iftex}

\ifpdf
  \usepackage{underscore}         
  \usepackage[T1]{fontenc}        
\else
  \usepackage{breakurl}           
\fi

\title{Operations on Boolean and Alternating Finite
	 Automata}
\author{Galina Jir\'askov\'a\footnote{This research was supported by the Slovak Grant Agency for Science (VEGA) under contract 2/0096/23 “Automata and Formal Languages: Descriptional and Computational Complexity”.}
\institute{Mathematical Institute, Slovak Academy of Sciences,  Gre\v{s}\'akova 6, 040 01 Ko\v{s}ice, Slovakia}
\email{jiraskov@saske.sk}
}
\def\titlerunning{Operations on Boolean and Alternating Finite
	Automata}
\def\authorrunning{G. Jir\'askov\'a}
\begin{document}
\maketitle

\begin{abstract}
We examine the complexity of basic regular operations 
on languages represented by  Boolean and alternating finite automata.  We get tight upper bounds $m+n$ and $m+n+1$
for union, intersection, and difference, 
$2^m+n$ and $2^m+n+1$ for concatenation,
$2^n+n$ and $2^n+n+1$ for square,
$m$ and $m+1$ for left quotient,
$2^m$ and $2^m+1$ for right quotient.
We also show that in both models,
the complexity of
complementation and symmetric difference is
$n$ and $m+n$, respectively,
while the complexity of star and reversal is $2^n$.
All our witnesses are described over a unary or binary alphabets,
and whenever we use a binary alphabet, it is 
always optimal.
\end{abstract}

\section{Introduction}

Boolean and alternating finite automata 
\cite{bl80,ck81,hjk18,ji12,ko76,kr20,le81}
are generalizations of nondeterministic
finite automata. They recognize regular languages, however, they
may be exponentially smaller, 
with respect to the  number of states, than equivalent
nondeterministic finite automata (NFAs).
While in an NFA the transition function 
maps any pair of a state and input symbol to a set of states
that can be viewed as a disjunction of the states,
in a Boolean finite automaton (BFA) the result of the transition function  is given
by any Boolean function with variables in the state set.

Fellah et al. \cite{fjy90} examined 
alternating finite automata (AFAs), that is, 
Boolean automata  in which the initial Boolean function 
is given by a projection. 
They proved that every $n$-state AFA  
can be simulated by a ($2^n+1$)-state nondeterministic finite automaton 
with a unique initial state, 
and left as an open problem the tightness of this upper bound. 
An answer to
this problem  was given in~\cite[Lemma~1, Theorem~1]{ji12}
by describing an $n$-state binary AFA 
whose equivalent NFA with a unique initial state 
has at least $2^n+1$ states.
Here we present a different example in which
the reachability and co-reachability of all singleton sets
immediately implies the result.

In~\cite{fjy90} it was also shown that 
given an~$m$-state and $n$-state AFAs for   languages~$K$ and~$L$,
the languages~$L^c$, $K\cup L$, $K\cap L$, $KL$, and~$L^*$
are recognized by AFAs of at most $n,m+n+1,m+n+1, 2^m+n+1$,
and~$2^n+1$ states, respectively,
and the tightness of these upper bounds was left open as well.

Here we present the results obtained in~\cite{hj18,hjk18,ji12,jk19,kr20}
that provide the exact complexity of basic regular operations
on languages represented by Boolean and alternating finite automata.
Table~\ref{ta1} summarizes these results. It also displays
the sizes of alphabet used to describe witness languages. 

\begin{table}[t]
	\renewcommand{\arraystretch}{1.2}  
\centering
\caption{The complexity of basic regular operations on Boolean and alternating finite automata.}
\vskip10pt
\begin{tabular}{@{}llllll@{}}
	\toprule
operation           & BFA & $|\Sigma|$ & AFA & $|\Sigma|$  & source\\
\midrule
complementation    & $n$ & $1$ & $n$ & $1$ & 
  \cite[Thm. 1]{hjk18}\\ 
union               & $m+n$ & 1 & $m+n+1$ & 1 &
  \cite[Thm. 2(1) and 3(1)]{ji12}, \cite[Thm. 4.3 and 4.4]{kr20}\\
intersection        & $m+n$ & 1 & $m+n+1$ & 1 & 
\cite[Thm. 2(2) and 3(2)]{ji12}, \cite[Thm. 4.3 and 4.4]{kr20}\\
difference          & $m+n$ & 1 & $m+n+1$ & 1 &
  \cite[Thm. 13(a) and 14(a)]{hjk18}, \cite[Thm. 4.3 and 4.4]{kr20}\\
symm. difference& $m+n$ & 1 & $m+n$ & 1 &
  \cite[Thm. 13(b) and 14(b)]{hjk18}, \cite[Thm. 4.3 and 4.4]{kr20} \\
star                & $2^n$  & 2 & $2^n$ & 2 & 
\cite[Thm. 12]{hjk18}\\
reversal            & $2^n$  & 2 & $2^n$ & 2 &
  \cite[Thm. 13(c) and 14(c)]{hjk18}\\
right quotient & $2^m$  & 2 & $2^m+1$ & 2 & 
\cite[Thm. 13(d) and 14(d) ]{hjk18}\\
left quotient       &  $m$  & 1 & $m+1$ & 1 & 
\cite[Thm. 13(e) and 14(e)]{hjk18}\\  
concatenation       & $2^m+n$ & 2 & $2^m+n+1$ & 2&
\cite[Thm. 4 and 5]{ji12},\cite[Thm. 6.4]{hj18}\\
square              & $2^n+n$ & 2 & $2^n+n+1$ & 2& 
\cite[Thm. 13 and 14]{jk19}\\  
\bottomrule
\end{tabular}
\label{ta1}
\end{table}

\section{Preliminaries}

Let~$\Sigma$ be a non-empty \emph{alphabet} of symbols.
Then~$\Sigma^*$ denotes the set of all \emph{strings} over 
the alphabet~$\Sigma$ including the empty string~$\eps$.
A \emph{language} over~$\Sigma$ is any subset of~$\Sigma^*$.

A \emph{Boolean finite automaton} (BFA) is 
a quintuple~$A=(Q,\Sigma,\cdot,g_s,F)$ 
where~$Q=\{q_1,q_2,\ldots,q_n\}$ is a finite non-empty
set of states,
$\Sigma$ is a finite input alphabet,
$\cdot$ is a transition function that maps~$Q\times\Sigma$
into the set~$\mathcal{B}_n$ of Boolean functions
with variables~$\{q_1,q_2,\ldots,q_n\}$, 
$g_s\in \mathcal{B}_n$ is the initial Boolean function,
and~$F\subseteq Q$ is the set of final states.
The transition function~$\cdot$ is extended 
to the domain~$\mathcal{B}_n\times\Sigma^*$
as follows: For each~$g\in\mathcal{B}_n$,
each~$a\in\Sigma$, and each~$w\in\Sigma^*$, we have
\begin{itemize}
	\item $g \cdot \varepsilon=g$,
	\item if~$g=g(q_1,q_2,\ldots,q_n)$,
	then~$g\cdot a = g(q_1\cdot a,q_2\cdot a,\ldots,q_n\cdot a)$,
	\item $g\cdot( wa) = (q\cdot w) \cdot a$.
\end{itemize}
Let~$f=(f_1,f_2,\ldots,f_n)$ be a Boolean vector (\emph{finality vector}) such that~$f_i=1$ if and only if~$q_i\in F$.
The \emph{language accepted by} a BFA~$A$ 
is the set of strings~$L(A)=\{w\in\Sigma^*\mid 
(g_s\cdot w)(f)=1\}$.
We illustrate the above mentioned notions in the following example.

\begin{example}\rm
Consider the 2-state binary Boolean finite 
automaton~$A=(\{q_1,q_2\},\{a,b\},\cdot, q_1\land q_2,\{q_1\})$
where the transition function~$\cdot$ 
is defined in Table~\ref{ta2}.
\begin{table}[h!]
	\centering
	\renewcommand{\arraystretch}{1.2}
	\caption{The transition function of the BFA~$A$.}
	\vskip10pt
	\begin{tabular}{{@{}lll@{}}}
		\toprule
		$\cdot$       & $a$                 & $b$ \\
		\midrule
		$q_1$         & $q_1\lor q_2$    & $q_1$ \\
		$q_2$         & $q_2$    & $q_1\land \lnot q_2$ \\
		\bottomrule
	\end{tabular}
\label{ta2}
\end{table}
\\
Then the string~$ab$ is accepted by~$A$ since we have
$$
 g_s\cdot  ab=  (q_1\land q_2) \cdot ab = ((q_1\lor q_2)\land q_2) \cdot b
   =( q_1\lor (q_1\land \lnot q_2)) \land (q_1\land \lnot q_2),
$$
and the resulting function evaluates 
to 1 in the finality vector~$(1,0)$.
\qed
\end{example}	
	
A BFA~$A$ is called \emph{alternating} (AFA)	
if its initial function 
is a projection~$g_s(q_1,q_2,\ldots,q_n)=q_1$;
cf.~\cite{ck81,fjy90,yu97}.
It is \emph{nondeterministic with multiple initial states} (MNFA)
if~$g_s$ and all~$q_i\cdot a$
are of the form~$q_{i_1}\lor q_{i_2} \lor \cdots\lor q_{i_\ell}$.
%
If moreover~$g_s=q_1$, then~$A$ is \emph{nondeterministic} 
(with a unique initial state) (NFA).
If moreover all~$q_i\cdot a$
are of the form~$q_j$, then~$A$ is \emph{deterministic} (DFA).

\section{Simulations of BFAs and AFAs 	by MNFAs, NFAs, and DFAs}
In this section we recall the trade-offs between different models of finite automata. Let us start with the simulation of BFAs by MNFAs.

\begin{proposition}[
	{\cite[Theorem~4.1]{fjy90}},
	{\cite[Lemma~1]{ji12}}]
	\label{prop:bfa-to-mnfa}
	Let~$L$ be a language accepted by an~$n$-state~BFA.
	Then~$L$ is accepted by a~$2^n$-state MNFA 
	whose reverse is a DFA.	
\end{proposition} 

\begin{proof}[Proof Idea]
	Let~$A=(Q,\Sigma,\cdot,g_s,F)$ be a BFA with~$Q=\{q_1,q_2,\ldots,q_n\}$.
	Let~$f=(f_1,f_2,\ldots,f_n)$ be the Boolean finality vector
	with~$f_i=1$ iff~$q_i\in F$.
	Construct a MNFA~$A'=(Q',\Sigma,\circ,I,\{f\})$ where
	\begin{itemize}
		\item $Q'=\{0,1\}^n$,
		\item $I=\{u\in Q'\mid g_s(u)=1\}$,
		\item for each~$u\in Q'$ and each~$a\in\Sigma$,
		we set~$u\circ a=\{u'\in Q'\mid (q_1\cdot a,q_2\cdot a,\ldots,q_n\cdot a)(u')=u\}$.
	\end{itemize}
Then~$L(A)=L(A')$.
\end{proof}

Since the reverse of the MNFA in the proof above is a DFA,
we get the next result. 

\begin{corollary}
	\label{co:bfa-to-dfa}
	If~$L$ is accepted by an~$n$-state BFA,
	then~$L^R$ is accepted by a~$2^n$-state DFA.
	\qed
\end{corollary}

Notice that if~$A$ is an~AFA, then the~MNFA~$A'$
constructed in the proof of Proposition~\ref{prop:bfa-to-mnfa}
has~$2^{n-1}$ initial states, and we get the following observation.

\begin{corollary}
	\label{co:afa-to-dfa}
	If~$L$ is accepted by an~$n$-state AFA,
	then~$L^R$ is accepted by a~$2^n$-state DFA
	of which~$2^{n-1}$ are final.
	\qed
\end{corollary}

Our next aim is to get the converses of the above corollaries.
\begin{proposition}[{\cite[Lemma~2]{ji12}}]
	Let~$L$ be accepted by a~$2^n$-state MNFA
	whose reverse is a DFA.
	Then~$L$ is accepted by an~$n$-state BFA.
\end{proposition}

\begin{proof}[Proof Idea]
	Let~$A=(Q,\Sigma,\cdot,I,F)$ be a~MNFA with~$Q=\{0,1\}^n$.
	Since~$A^R$ is a DFA, the MNFA~$A$ has a unique final state~$f\in Q$, and moreover, 
	for each~$u\in Q$ and each~$a\in\Sigma$
	there is a unique state~$u'$ with~$u'\cdot a=u$;
	denote this state by~$au$. 
	Construct a BFA~$A'=(Q',\Sigma,\circ,g_s,F')$ where
	\begin{itemize}
		\item $Q'=\{q_1,q_2,\ldots,q_n\}$,
		\item $g_s(u)=1$ iff~$u\in I$,
		\item $F'=\{q_i\mid f_i=1\}$,
		\item $(q_1\circ a,q_2\circ a, \ldots,q_n\circ a)(u)=au$.
	\end{itemize}
Then~$L(A)=L(A')$.
\end{proof}

\begin{corollary}
	\label{co:dfa-to-bfa}
	If~$L$ is accepted by a~$2^n$-state DFA,
	then~$L^R$ is accepted by an~$n$-state BFA.
	\qed
\end{corollary}

\begin{corollary}
	\label{co:dfa-to-afa}
	If~$L$ is accepted by a~$2^n$-state DFA
    which has~$2^{n-1}$ final states,
	then~$L^R$ is accepted by an~$n$-state AFA.
	\qed
\end{corollary}

We continue with the simulation of BFAs by DFAs.	
	
\begin{proposition}[
	{\cite[Theorem~7]{ko76}},
	{\cite[Theorem~2]{bl80}},
	{\cite[Theorem~5.2]{ck81}},
	{\cite[Corollary~3]{le81}}]
	Let~$L$ be a language over an alphabet~$\Sigma$ 
	accepted by an~$n$-state BFA.
	Then~$L$ is accepted by a DFA of at most~$2^{2^n}$ states,
	and this upper bound is tight if~$|\Sigma|\ge2$.
\end{proposition}

\begin{proof}[Proof Idea]
	If~$L$ is accepted by an~$n$-state BFA,
	then by Proposition~\ref{prop:bfa-to-mnfa},
	it is accepted by a~$2^n$-state MNFA,
	and, consequently, by a~$2^{2^n}$-state DFA.
	For tightness, let~$K$ be the binary~$2^n$-state DFA
	from~\cite[Proposition~2]{le81}
	whose reversal~$K^R$ requires~$2^{2^n}$ deterministic states.
	By Corollary~\ref{co:dfa-to-bfa}, the language~$K^R$
	is accepted by an~$n$-state BFA.
\end{proof}
Finally, we consider the simulation of BFAs by NFAs,
and provide an answer to an open problem from~\cite{fjy90}.

\begin{theorem}[{\cite[Theorem~1]{ji12}}]
	Let~$L$ be 
	accepted by an~$n$-state BFA.
	Then~$L$ is accepted by an NFA of at most~$2^n+1$ states.
	This upper bound is tight, 
	and it can be met by a binary~$n$-state AFA.
\end{theorem}

\begin{proof}[Proof Idea]
	By Proposition~\ref{prop:bfa-to-mnfa},
	the language~$L$ is accepted by a~$2^n$-state MNFA,
	and, consequently, by a~$(2^n+1)$-state NFA.
	For tightness, let~$n\ge2$. 
	Let~$L$ be the language accepted by the~$2^n$-state 
	MNFA $A=(Q,\{a,b\},\cdot,I,F)$ where
	\begin{itemize}
		\item $Q=\{0,1,\ldots,2^n-1\}$,
		\item $I=\{0,1,\ldots,2^{n-1}-1\}$,
		\item $F=\{2^n-1\}$,
		\item $i\cdot a =\{(i+1)\bmod 2^n\}$ for each~$i\in Q$,
		\item $0\cdot b=\{0\}$, $(2^n-1)\cdot b= Q\setminus\{0\}$,
		and~$i\cdot b=\emptyset$ is~$i\in Q\setminus\{0,2^n-1\}$;		
	\end{itemize}
    see Figure~\ref{fig:mnfaA} for an illustration~\ref{fig:mnfaA}.
    	The reverse~$A^R$ is a~$2^n$-state DFA which has~$2^{n-1}$ final states.  
    By Corollary~\ref{co:dfa-to-afa}, the language~$L$
    is accepted by an~$n$-state AFA. 
    On the other hand, each singleton set
    is reachable and co-reachable in the MNFA~$A$
    which means that
    every NFA accepting~$L$ has at least~$2^n+1$ states
    by~\cite[Lemma~9]{ho21}.	
\end{proof}		

\begin{figure}[h] 
	\centering
	\begin{tikzpicture}[>=stealth', initial text={},shorten 
	>=1pt,auto,inner sep=2,node distance=1.9cm]
	\node[state,initial above] (0) [label=center:{$0$}]{};
	\node[state,initial above] (1) [right of=0,label=center:{$1$}]{};
	\node[state,initial above] (2) [right of=1,label=center:{$2$}]{};
	\node[state,initial above] (3) [right of=2,label=center:{$3$}]{};
	\node[state] (4) [right of=3,label=center:{$4$}]{};
	\node[state] (5) [right of=4,label=center:{$5$}]{};
	\node[state] (6) [right of=5,label=center:{$6$}]{};
	\node[state,accepting] (7) [right of=6,label=center:{$7$}]{};
			
	\draw [->] (0) to node   {$a$} (1);
	\draw [->] (1) to node   {$a$} (2);
	\draw [->] (2) to node   {$a$} (3);
	\draw [->] (3) to node   {$a$} (4);
	\draw [->] (4) to node   {$a$} (5);
	\draw [->] (5) to node   {$a$} (6);
	\draw [->] (6) to node   {$a$} (7);
	\draw [->] (7.250) to  [bend left,above,pos=0.7] node  {$a$} (0);
	\draw [->,red] (0) to [loop left ] node   {$b$} (0);
	\draw [->,red] (7) to [loop above ] node   {$b$} (7);
	\draw [->,red] (7.235) to  [bend left,above,pos=0.7] node  {$b$} (1);
	\draw [->,red] (7.230) to  [bend left,above,pos=0.7] node  {$b$} (2);
	\draw [->,red] (7.225) to  [bend left,above,pos=0.7] node  {$b$} (3);
	\draw [->,red] (7.220) to  [bend left,above,pos=0.7] node  {$b$} (4);
	\draw [->,red] (7.215) to  [bend left,above,pos=0.7] node  {$b$} (5);
	\draw [->,red] (7) to  [bend left,above] node  {$b$} (6);	
	\end{tikzpicture}
	\caption{The  MNFA $A$; $n=3$.}
	\label{fig:mnfaA}
\end{figure}
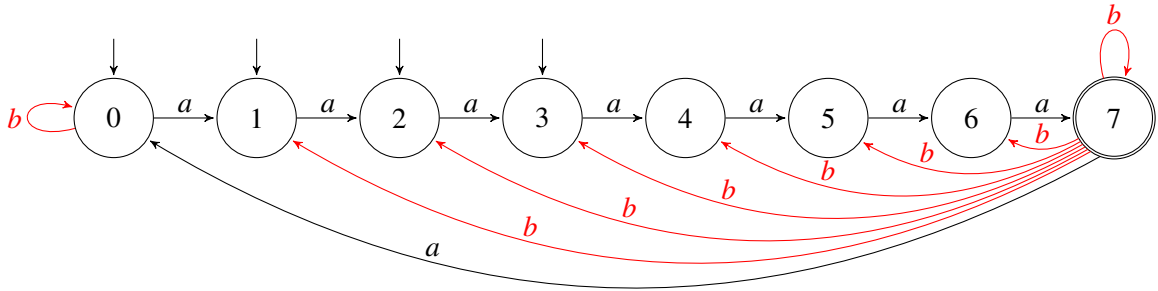

\section{Operational Complexity on Boolean and Alternating Finite Automata}

In this section  we use the four corollaries from the previous section
to get the complexity of basic regular operations on languages
represented by Boolean and alternating finite automata.
The idea is as follows. 
Consider a binary operation and take languages~$K$ and~$L$
recognized by a~$2^m$-state and~$2^n$-state DFA, respectively,
that are witnesses for the considered operation on DFAs.
Then the languages~$K^R$ and~$L^R$
are accepted by an~$m$-state and~$n$-state BFA, respectively.
Now it is enough to show that the language resulting 
from the operation applied to the languages~$K^R$ and~$L^R$
requires large enough BFA. In the case of AFAs,
we start with DFAs with half of their states final
that are hard for the considered operation on DFAs.
We illustrate this idea for the concatenation operation.

\begin{theorem}[\bf Concatenation on BFAs]
	Let~$K$ and~$L$ be languages over an alphabet~$\Sigma$
	accepted by an~$m$-state and~$n$-state BFA, respectively.
	Then the language~$KL$ is accepted by a BFA of at most~$2^m+n$ states,
	and this upper bound is tight if~$|\Sigma|\ge2$.
\end{theorem}

\begin{proof}
	To get an upper bound, let~$A=(Q_A,\Sigma,\cdot_A,g_A,F_A)$ 
	and~$B=(Q_B,\Sigma,\cdot_B,g_B,F_B)$ be BFAs accepting
	the languages~$K$ and~$L$, respectively.
	We first convert the BFA $A$
	to the~$2^m$-state MNFA $M=(Q_M,\Sigma,\cdot_M,g_M,F_M)$.
	Now we construct a BFA $C=(Q_M\cup Q_B,\Sigma,\cdot,g_M,F_B)$ with
	$$
	q\cdot a =
	\begin{cases}
		q\cdot_M a,                    & \text{if $q\in Q_M\setminus F_M$};\\
		q\cdot_M a \lor g_B \cdot_B a, & \text{if $q\in F_M$};\\
		q\cdot_B a,                    & \text{if $q\in Q_B$};\\
	\end{cases}
   $$
   cf.~\cite[Theorem~9.2]{fjy90}. Then the BFA~$C$ has~$2^m+n$ states 
   and recognizes the language~$KL$.
   
   To get tightness, let~$K$ and~$L$ be Maslov's binary witnesses for
   concatenation on DFAs from~\cite{ma70}, see Figure~\ref{fig:ma},
   accepted by a~$2^n$-state and~$2^m$-state DFA, respectively.
   Then every DFA accepting the language~$KL$ has at least~$2^n2^{2^m}-2^{2^m-1}$ states.
   By Corollary~\ref{co:dfa-to-bfa}, the languages~$L^R$ and~$K^R$ 
   are accepted by~$m$-state and~$n$-state BFA,    respectively.
   Next, we have~$(L^RK^R)^R=KL$,
   so every DFA accepting the reverse of the concatenation~$L^RK^R$ has at least~$2^n2^{2^m}-2^{2^m-1}$ states.
   By Corollary~\ref{co:bfa-to-dfa}, it follows that every BFA accepting~$K^RL^R$
   has at least~$\lceil\log(2^n2^{2^m}-2^{2^m-1})\rceil=2^m+n$ states.
\end{proof}

\begin{figure}[h] 
		\hglue1.5cm
	\begin{tikzpicture}[>=stealth', initial text={$A$},shorten 
			>=1pt,auto,inner sep=2,node distance=1.9cm]
	\node[state,initial] (0) [label=center:{$1$}]{};
	\node[state] (1) [right of=0,label=center:{$2$}]{};
	\node[state,draw=none] (2) [right of=1,label=center:{$\ldots$}]{};
	\node[state] (3) [right of=2,label=center:{$m-1$}]{};
	\node[state,accepting] (4) [right of=3,label=center:{$m$}]{};
			
	\draw [->] (0) to node   {$a$} (1);
	\draw [->] (1) to node   {$a$} (2);
	\draw [->] (2) to node   {$a$} (3);
	\draw [->] (3) to node   {$a$} (4);
	\draw [->] (4) to [bend left=25,above] node   {$a$} (0);
	\draw [->] (0) to [loop above] node   {$b$} (0);
	\draw [->] (1) to [loop above] node   {$b$} (1);
	\draw [->] (3) to [loop above] node   {$b$} (3);
	\draw [->] (4) to [loop above] node   {$b$} (4);
\end{tikzpicture}	
\\
\hglue4cm
	\begin{tikzpicture}[>=stealth', initial text={$B$},shorten 
	>=1pt,auto,inner sep=2,node distance=1.9cm]
	\node[state,initial] (0) [label=center:{$1$}]{};
	\node[state] (1) [right of=0,label=center:{$2$}]{};
	\node[state,draw=none] (2) [right of=1,label=center:{$\ldots$}]{};
	\node[state] (3) [right of=2,label=center:{$n-2$}]{};
	\node[state] (4) [right of=3,label=center:{$n-2$}]{};
	\node[state,accepting] (5) [right of=4,label=center:{$n$}]{};
	
	\draw [->] (0) to node   {$b$} (1);
	\draw [->] (1) to node   {$b$} (2);
	\draw [->] (2) to node   {$b$} (3);
	\draw [->] (3) to node   {$b$} (4);
	\draw [->] (4) to node   {$a,b$} (5);
	\draw [->] (5) to [loop above] node   {$b$} (5);
	
	\draw [->] (0) to [loop above] node   {$a$} (0);
	\draw [->] (1) to [loop above] node   {$a$} (1);
	\draw [->] (3) to [loop above] node   {$a$} (3);
	\draw [->] (5) to [bend left,above] node   {$a$} (4);
		\end{tikzpicture}
	\caption{Maslov's witness DFAs for concatenation meeting the upper bound~$m2^n-2^{n-1}$.}
	\label{fig:ma}
\end{figure}
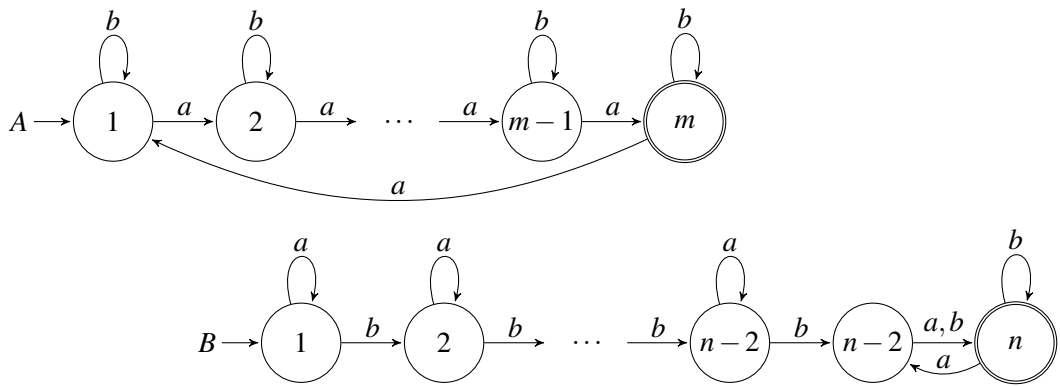

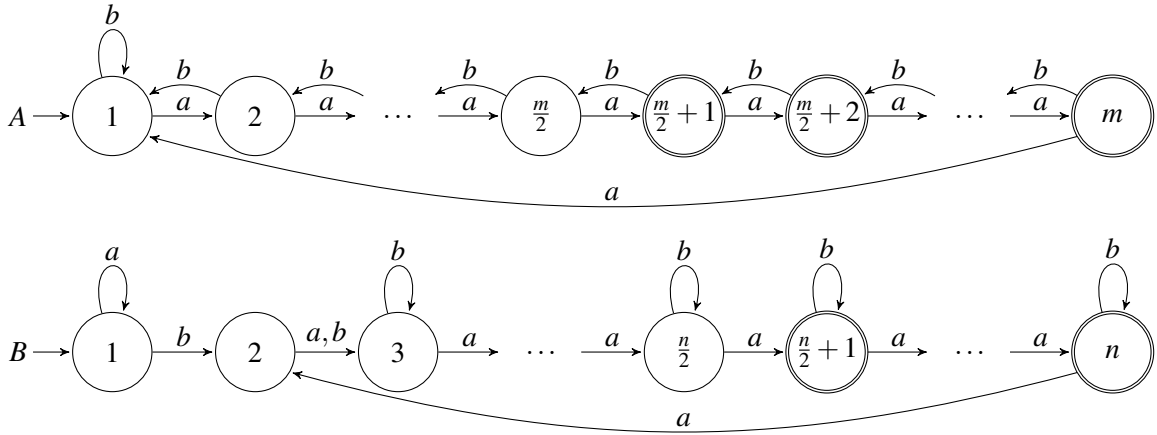
\begin{figure}[h] 
		\begin{tikzpicture}[>=stealth', initial text={$A$},shorten 
			>=1pt,auto,inner sep=2,node distance=1.9cm]
			\node[state,initial] (0) [label=center:{$1$}]{};
			\node[state] (1) [right of=0,label=center:{$2$}]{};
			\node[state,draw=none] (2) [right of=1,label=center:{$\ldots$}]{};
			\node[state] (3) [right of=2,label=center:{$\frac{m}{2}$}]{};
			\node[state,accepting] (4) [right of=3,label=center:{$\frac{m}{2}+1$}]{};
			\node[state,accepting] (5) [right of=4,label=center:{$\frac{m}{2}+2$}]{};
			\node[state,draw=none] (6) [right of=5,label=center:{$\ldots$}]{};
			\node[state,accepting] (7) [right of=6,label=center:{$m$}]{};
			
			\draw [->] (0) to node   {$a$} (1);
			\draw [->] (0) to [loop above] node   {$b$} (0);
			\draw [->] (1) to node   {$a$} (2);
			\draw [->] (2) to node   {$a$} (3);
			\draw [->] (3) to node   {$a$} (4);
			\draw [->] (4) to node   {$a$} (5);
			\draw [->] (5) to node   {$a$} (6);
			\draw [->] (6) to node   {$a$} (7);
			\draw [->] (7.210) to [bend left=15,above] node   {$a$} (0.-30);
			\draw [->] (1) to [bend right,above] node   {$b$} (0);
			\draw [->] (2) to [bend right,above] node   {$b$} (1);
			\draw [->] (3) to [bend right,above] node   {$b$} (2);
			\draw [->] (4) to [bend right,above] node   {$b$} (3);
			\draw [->] (5) to [bend right,above] node   {$b$} (4);
			\draw [->] (6) to [bend right,above] node   {$b$} (5);
			\draw [->] (7) to [bend right,above] node   {$b$} (6);			
	\end{tikzpicture}	
		\\
		\begin{tikzpicture}[>=stealth', initial text={$B$},shorten 
			>=1pt,auto,inner sep=2,node distance=1.9cm]
			\node[state,initial] (0) [label=center:{$1$}]{};
			\node[state] (1) [right of=0,label=center:{$2$}]{};
			\node[state] (2) [right of=1,label=center:{$3$}]{};
			\node[state,draw=none] (3) [right of=2,label=center:{$\ldots$}]{};
			\node[state] (4) [right of=3,label=center:{$\frac{n}{2}$}]{};
			\node[state,accepting] (5) [right of=4,label=center:{$\frac{n}{2}+1$}]{};
			\node[state,draw=none] (6) [right of=5,label=center:{$\ldots$}]{};
			\node[state,accepting] (7) [right of=6,label=center:{$n$}]{};
			
			\draw [->] (0) to node   {$b$} (1);
			\draw [->] (1) to node   {$a,b$} (2);
			\draw [->] (2) to [loop above] node   {$b$} (2);
			\draw [->] (4) to [loop above] node   {$b$} (4);
			\draw [->] (5) to [loop above] node   {$b$} (5);
			\draw [->] (7) to [loop above] node   {$b$} (7);

			\draw [->] (0) to [loop above] node   {$a$} (0);
			\draw [->] (2) to   node   {$a$} (3);
			\draw [->] (3) to   node   {$a$} (4);
			\draw [->] (4) to   node   {$a$} (5);
			\draw [->] (5) to   node   {$a$} (6);
			\draw [->] (6) to   node   {$a$} (7);
			\draw [->] (7.210) to [bend left=15,above] node   {$a$} (1.-30);
		\end{tikzpicture}
		\caption{Witness DFAs for concatenation with half of states final meeting the upper bound~$m2^n-\frac{m}{2}2^{n-1}$.}
	\label{fig:hos}
\end{figure}

\begin{theorem}[\bf Concatenation on AFAs]
	Let~$K$ and~$L$ be languages over an alphabet~$\Sigma$
	accepted by an~$m$-state and~$n$-state AFA, respectively.
	Then the language~$KL$ is accepted by a AFA of at most~$2^m+n+1$ states,
	and this upper bound is tight if~$|\Sigma|\ge2$.
\end{theorem}

\begin{proof}
	The upper bound follows from the previous theorem
	since one more state is enough to get an AFA equivalent to a given BFA.
	To get tightness, we use languages~$K$ and~$L$ accepted 
	by~$2^n$-state and~$2^m$-state witness DFAs for concatenation
	with half of their states final from~\cite[Theorem~4.7]{hj18},
	see Figure~\ref{fig:hos}. Then the minimal  DFA for~$KL$
	has~$2^n2^{2^m}-2^{n-1}2^{2^m-1}$ states,
	of which more that~$2^{2^m+n-1}$ states are final~\cite[Lemma~6.4]{hj18}.
	Then the languages~$L^R$ and~$K^R$ are accepted
	by an~$m$-state and~$n$-state AFA, respectively.
	Next, we have~$(L^RK^R)^R=KL$, so
	every AFA for~$L^RK^R$ has at least~$\lceil\log(2^n2^{2^m}-2^{n-1}2^{2^m-1})\rceil=2^m+n$
	states. If an AFA of~$2^m+n$ states would accept~$L^RK^R$,
	then the reverse of this language, that is, the language~$KL$
	would be accepted by a DFA of~$2^{2^m+n}$ states with~$2^{2^m+n-1}$ final states. However,
	the minimal DFA for~$KL$ has more than~$2^{2^m+n-1}$ final states, a contradiction.  
\end{proof}

Hence, the upper bound~$2^m+n+1$
for concatenation on AFAs from~\cite[Theorem~9.3]{fjy90}
is tight. This provides an answer to the second open problem from~\cite{fjy90}. A similar idea as for concatenation
also works for square, and left and right quotients.
Our results for the star operation are covered by the next theorem.

\begin{theorem}[\bf Star on BFAs and AFAs]
	If~$L$ is accepted by an~$n$-state BFA,
	then~$L^*$ is accepted by a~$2^n$-state AFA.
	Moreover, there exists a binary language~$L$
	accepted by an~$n$-state AFA
	such that every BFA for~$L^*$ has at least~$2^n$ states.
\end{theorem}

\begin{proof}
	If~$L$ is accepted by an~$n$-state BFA,
	then~$L^R$ is accepted by a~$2^n$-state DFA by Corollary~\ref{co:bfa-to-dfa}.
	Then~$(L^R)^*$ is accepted by a~$2^{2^n}$-state DFA 
	with half of its state final~\cite[Proposition~8]{hjk18}.
	Next, we have~$(L^*)^R=(L^R)^*$.
	Hence~$L^*$ is accepted by an~$n$-state AFA by Corollary~\ref{co:dfa-to-afa}.
	
	To get tightness, let~$L$ be the Palmovsk\'y's
	witness DFA for star with~$2^n$ states half of which are final~\cite[Theorem~4.4]{pa16}, see Figure~\ref{fig:pal}.
	Then~$L^R$ is accepted by an~$n$-state AFA by Corollary~\ref{co:dfa-to-afa}.
	Next, we have~$((L^R)^*)^R=((L^*)^R)^R=L^*$,
	and every DFA for~$L^*$ has at least~$2^{2^n-1}+2^{2^n-1+2^{n-1}}$
	states. It follows that every BFA for~$(L^R)^*$
	has at least~$\lceil\log(2^{2^n-1}+2^{2^n-1+2^{n-1}})\rceil=2^n$
	states by Corollary~\ref{co:bfa-to-dfa}.
\end{proof}

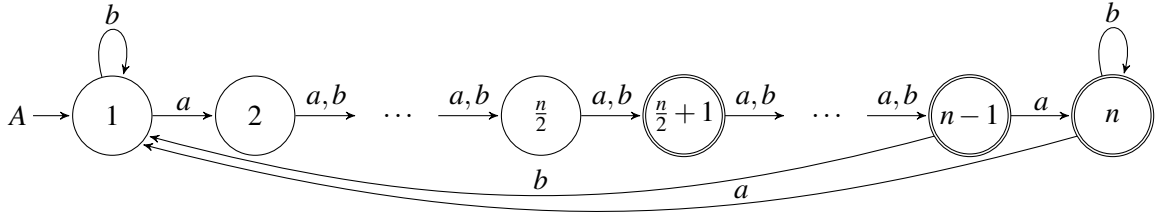
\begin{figure}[h] 
		\begin{tikzpicture}[>=stealth', initial text={$A$},shorten 
			>=1pt,auto,inner sep=2,node distance=1.9cm]
			\node[state,initial] (0) [label=center:{$1$}]{};
			\node[state] (1) [right of=0,label=center:{$2$}]{};
			\node[state,draw=none] (2) [right of=1,label=center:{$\ldots$}]{};
			\node[state] (3) [right of=2,label=center:{$\frac{n}{2}$}]{};
			\node[state,accepting] (4) [right of=3,label=center:{$\frac{n}{2}+1$}]{};			
			\node[state,draw=none] (5) [right of=4,label=center:{$\ldots$}]{};
			\node[state,accepting] (6) [right of=5,label=center:{$n-1$}]{};
			\node[state,accepting] (7) [right of=6,label=center:{$n$}]{};
			
			\draw [->] (0) to node   {$a$} (1);
			\draw [->] (1) to node   {$a,b$} (2);
			\draw [->] (2) to node   {$a,b$} (3);
			\draw [->] (3) to node   {$a,b$} (4);
			\draw [->] (4) to node   {$a,b$} (5);
			\draw [->] (5) to node   {$a,b$} (6);
			\draw [->] (6) to node   {$a$} (7);
			\draw [->] (7.210) to [bend left=15,above,pos=0.35] node   {$a$} (0.-45);
			\draw [->] (0) to [loop above] node   {$b$} (0);
			\draw [->] (7) to [loop above] node   {$b$} (7);
 			\draw [->] (6.210) to [bend left=15,above] node   {$b$} (0.-30);	
		\end{tikzpicture}	
\caption{Witness DFA for star with half of states final meeting the upper bound~$2^{n-1}+2^{n-1-\frac{n}{2}}$.}
\label{fig:pal}
\end{figure}

Similar arguments work for reversal.
If~$L$ is accepted by an~$n$-state BFA,
that~$L^R$ is accepted by a~$2^n$-state DFA, a special case of~AFA.
For tightness, we take the language~$L$ accepted by a~$2^n$-state
\v{S}ebej's DFA from~\cite[Fig.~6]{js12} 
with half of its states final. Then~$L^R$
ia accepted by an~$n$-state AFA,
while every DFA for~$L^R$
has at least~$2^{2^n}$ states.
Hence, every BFA for~$L=(L^R)^R$ has at least~$2^n$ states
by Corollary~\ref{co:dfa-to-bfa}.

We conclude this section with Boolean operations.
Denote by~bsc$(L)$ the number of states in a minimal,
with respect to the number of states, BFA accepting~$L$.
Define~asc$(L)$ in an analogous way.

\begin{proposition}
	\label{prop:complement}
	Let~$L$ be a regular language.
	Then~bsc$(L)=$bsc$(L^c)$
    and~asc$(L)=$asc$(L^c)$.
\end{proposition}

\begin{proof}
	If~$L$ is accepted by a minimal~$n$-state BFA,
	then~$L^R$ is accepted by a~$2^n$-state DFA by Corollary~\ref{co:bfa-to-dfa}.
	It follows that~$(L^R)^c=(L^c)^R$
	is accepted by a~$2^n$ state DFA,
	and therefore~$L^c$ is accepted by an~$n$-state BFA
	by Corollary~\ref{co:dfa-to-bfa}.
	Moreover, the language~$L^c$
	cannot be accepted by a smaller BFA
	because otherwise the language~$L=(L^c)^c$
	would be accepted by a smaller BFA as well.
	In the case of AFAs,
	the DFAs for~$L^R$ and~$(L^R)^c$
	have~$2^n$ states and~$2^{n-1}$ final states,
	and we use Corollaries~\ref{co:afa-to-dfa} and~\ref{co:dfa-to-afa}
	to get the result.
\end{proof}

\begin{theorem}
	Let~$K$ and~$L$ be languages over~$\Sigma$
	accepted by an~$m$-state and~$n$-state AFA, respectively.
	Then~$K\cup L$
	is accepted by an AFA of at most~$m+n+1$ states,
	and this upper bound is tight if~$|\Sigma|\ge1$.
\end{theorem}

\begin{proof}
	The language~$K\cup L$ can be accepted by a~$(m+n)$-state BFA
	constructed from the two AFAs by setting the initial function
	to the disjunction of the corresponding initial states.
	The upper bound for AFAs follows.
	For tightness, let~$K$ be the language
	accepted by the unary~$2^m$-state DFA
	with  state set~$\{0,1,\ldots,2^m-1\}$,
	the initial state 0, the set of 
	final states~$\{2^{m-1},2^{m-1}+1,\ldots ,2^m-1\}$,
	and transitions given by~$i\cdot a =(i+1)\bmod 2^m$.
	Then~$K^R=K$ is accepted by an~$m$-state AFA.
	Next, let~$L$ be a language accepted by a~$(2^n-1)$-state unary DFA 	with  state set~$\{0,1,\ldots,2^n-2\}$,
	the initial state 0, the set of 
	final states~$\{2^{n-1},2^{m-1}+1,\ldots ,2^m-2\}$,
	and transitions given by~$i\cdot a =(i+1)\bmod 2^m-1$.
	Then we can add an unreachable final state to this DFA
	to get an equivalent~$2^n$-state DFA with half of its states final. Hence~$L^R=L$ is accepted by an~$n$-state AFA.
	As shown in~\cite[Lemma~4.2, Theorem~4.4]{kr20},
	the minimal DFA for~$K\cup L$ has~$2^m(2^n-1)$ states,
	of which more than~$2^{m+n+1}$ are final.
	It follows that every AFA for~$K\cup L$  
	has at least~$m+n+1$ states.
\end{proof}

By Proposition~\ref{prop:complement} and De Morgan's laws,
the complement of the languages described in the previous proof are witnesses for intersection. The case of difference is analogous.
The same languages give a lower bound~$m+n$ 
for symmetric difference on AFAs~\cite[Lemma~4.2]{kr20} 
which is also an upper bound;
notice that the symmetric difference of two DFAs with half of their states final is accepted by a DFA with half of its states final.
Finally, exactly the same languages serve as witnesses for Boolean operations on BFAs~\cite[Theorem~4.3]{kr20}.

In the unary case, the reverse of any language is the same language,
and the right quotient is the same as the left quotient of the corresponding languages.
Moreover, we can show that the complexity of  star, concatenation,
and square on unary BFAs is~$2n,m+n,$ and~$n+1$, respectively.
It follows that whenever we used a binary alphabet to describe 
witnesses for the corresponding operations on BFAs and AFAs,
it was always optimal. 
 
The exact complexity of star, concatenation, and square on unary AFAs
remains open since the complexity of these operations on unary DFAs
with half of their states final is not known.
The complexity of less common regular operations
like shuffle, cyclic shift, or square root,
would be of interest as well.

%

\bibliographystyle{eptcs}
\bibliography{bfa}
\end{document}